\documentclass[11pt]{article}

\newif\ifcomments
\commentstrue
\usepackage{fullpage}
\usepackage[utf8]{inputenc}

\usepackage{graphicx}
\usepackage{amsmath,amsthm,amssymb}
\usepackage{algorithm}
\usepackage{algpseudocode}
\usepackage{multirow}
\usepackage{makecell}
\usepackage{algpseudocode}

\usepackage{thm-restate,color,xspace}
\usepackage{comment}
\usepackage{thmtools}
\usepackage{xcolor}
\usepackage{nameref}
\usepackage{array}
\usepackage{mathrsfs}

\definecolor{ForestGreen}{rgb}{0.1333,0.5451,0.1333}
\definecolor{DarkRed}{rgb}{0.65,0,0}
\definecolor{Red}{rgb}{1,0,0}
\usepackage[linktocpage=true,
pagebackref=true,colorlinks,
linkcolor=DarkRed,citecolor=ForestGreen,
bookmarks,bookmarksopen,bookmarksnumbered]
{hyperref}
\usepackage{cleveref}

\usepackage{xcolor}
\usepackage{nameref}

\usepackage{cleveref}

\makeatother

\usepackage{babel}
\usepackage{hyperref}

\declaretheorem[numberwithin=section]{theorem}
\declaretheorem[numberlike=theorem]{lemma}

\declaretheorem[numberlike=theorem]{claim}

\declaretheorem[numberlike=theorem,style=definition]{definition}
\declaretheorem[numberlike=theorem,style=definition]{remark}

\global\long\def\polylog{\mathrm{polylog}}

\global\long\def\poly{\mathrm{poly}}

\global\long\def\H{\mathcal{H}}
\global\long\def\vol{\mathrm{vol}}

\global\long\def\route{\mathrm{route}}

\newcommand{\ignore}[1]{}

\usepackage{todonotes}

\ifcomments
\def\thatchaphol#1{\marginpar{$\leftarrow$\fbox{TS}}\footnote{$\Rightarrow$~{\sf\textcolor{purple}{#1 --Thatchaphol}}}}
\def\benyu#1{\marginpar{$\leftarrow$\fbox{BW}}\footnote{$\Rightarrow$~{\sf\textcolor{red}{#1 --Benyu}}}}
\def\yaowei#1{\marginpar{$\leftarrow$\fbox{YL}}\footnote{$\Rightarrow$~{\sf\textcolor{blue}{#1 --Yaowei}}}}

\else
\newcommand{\benyu}[1]{}
\newcommand{\yaowei}[1]{}
\newcommand{\thatchaphol}[1]{}
\fi 

\begin{document}

\title{Approximating Directed Minimum Cut and Arborescence Packing via Directed Expander Hierarchies}
\date{\today}
\author{
Yonggang Jiang\thanks{MPI-INF and Saarland University, \texttt{yjiang@mpi-inf.mpg.de}.}
\and
Yaowei Long\thanks{University of Michigan, \texttt{yaoweil@umich.edu}.}
\and
Thatchaphol Saranurak\thanks{
        University of Michigan,
        \texttt{thsa@umich.edu}.
        Supported by NSF Grant CCF-2238138.}
\and
Benyu Wang\thanks{
University of Michigan, \texttt{benyuw@umich.edu}.
}
}
\maketitle

\begin{abstract}
   We give almost-linear-time algorithms for approximating rooted minimum cut and maximum arborescence packing in directed graphs, two problems that are dual to each other~\cite{edmonds1973edge}. More specifically, for an $n$-vertex, $m$-edge directed graph $G$ whose $s$-rooted minimum cut value is $k$,
\begin{enumerate}
\item Our first algorithm computes an $s$-rooted cut of size at most $O(k\log^{5} n)$ in $m^{1+o(1)}$ time, and
\item Our second algorithm packs $k$ $s$-rooted arborescences with $n^{o(1)}$ congestion in $m^{1+o(1)}$ time, certifying that the $s$-rooted minimum cut is at least $k / n^{o(1)}$.
\end{enumerate}
Our first algorithm works for weighted graphs as well.

Prior to our work, the fastest algorithms for computing the $s$-rooted minimum cut were exact but required super-linear running time: either $\tilde{O}(mk)$~\cite{gabow1991matroid} or $\tilde{O}(m^{1+o(1)}\min\{\sqrt{n},n/m^{1/3}\})$ \cite{cen2022minimum}. The fastest known algorithms for packing $s$-rooted arborescences had no congestion but required $\tilde{O}(m \cdot \mathrm{poly}(k))$ time~\cite{bhalgat2008fast}. 

\end{abstract}

\section{Introduction}
The (global) \emph{minimum cut} problem is very well-studied in graph algorithms. In the context of directed graphs, the goal is to find a cut $(S,V-S)$ in the graph $G=(V,E)$ such that the number of edges from $S$ to $V-S$ is minimized. Given an arbitrary source $s \in V$, the \emph{$s$-rooted minimum cut} is a cut $(S,V\setminus S)$ of minimum size with $s \in S$. The size of the $s$-rooted minimum cut is called $s$-rooted connectivity.
Observe that the global minimum cut problem can be solved by running two $s$-rooted minimum cut computations, one on the original graph and one on the graph with all edges reversed.

There is a line of work on solving the \emph{$s$-rooted minimum cut} problem.
Hao and Orlin \cite{hao1994faster} showed that a $\tilde{O}(mn)$-time algorithm based on push-relabel methods  \cite{goldberg1988new}. Then, Gabow \cite{gabow1991matroid} gave an algorithm with $\tilde{O}(mk)$ time where $k$ is the size of the $s$-rooted minimum cut.
A recent work \cite{cen2022minimum} solves the problem in $\min(n/m^{1/3},\sqrt{n}) \cdot m^{1+o(1)}$ time.

A closely related object to $s$-rooted mincut is an \emph{$s$-rooted arborescence}.
An $s$-rooted arborescence is a directed tree rooted at $s$ where all arrows point away from the root $s$. 
A very influential result by Edmonds \cite{edmonds1973edge} shows that the maximum number of disjoint $s$-rooted arborescences equals the size of $s$-rooted minimum cut.

In the \emph{arborescence packing} problem, we need to compute the maximum number of disjoint arborescences rooted at $s$. Gabow \cite{gabow1991matroid} an $O(n^2k^2)$-time algorithm. Later, Bhalgat et al.~\cite{bhalgat2008fast} improved the running time for arborescence packing to $\tilde{O}(m \poly(k))$. In weighted graphs, $k$ can be very large. In this case, the fastest algorithm is by Gabow and Manu \cite{gabow1998packing} and takes $\tilde{O}(mn^2)$ time.

In short, the state-of-the-art algorithms for both problems still require super-linear time. It was unclear if a speed-up is possible even when we allow approximation.

\subsection{Our Results}

In this paper, we present the first almost-linear-time algorithms for both problems, via directed expander hierarchies. 

For the $s$-rooted minimum cut problem, we give an algorithm that achieves with high probability an $O(\log^5(n))$-approximation.

\begin{restatable}{theorem}{cut}
\label{thm:cut}
Let $G = (V,E)$ be a directed graph with a given source vertex $s$ where the $s$-rooted connectivity is $k$.  
There is an algorithm that 
computes a cut $(S,V\setminus S)$ where $s\in S$ of size  $O(k \cdot \log^5(n))$. The running time is $m^{1+o(1)}$.
\end{restatable}
\Cref{thm:cut} generalizes seamlessly to weighted graphs as well.

Before stating our next result, we first define the \emph{congestion} of  a set ${\cal T}=\{T_1,\dots,T_{k'}\}$ of $s$-rooted arborescences as the maximum number of times any edge may appear in $T_i \in {\cal T}$, i.e., $\max_{e \in E}|\{ T_i \in {\cal T} \mid e\in T_i\}|$. 
Our second algorithm achieves $n^{o(1)}$ congestion in almost-linear time.\footnote{We note that our proof of \Cref{thm:pack} (see \Cref{thm:ArboPackingDetailed} for a detailed version) assumes an almost-linear-time directed expander routing subroutine (see \Cref{conj:ExpanderRouting}), which is a commonly believed, but we are not aware of a formal proof in the existing literature. We will work on a full proof in the next version.}

\begin{restatable}{theorem}{pack}
\label{thm:pack}
Let $G = (V,E)$ be a directed graph with a given source vertex $s$ and a parameter $k>0$. There is an algorithm that computes either
\begin{itemize}
\item  a cut $(S,V\setminus S)$ where $s\in S$ of size less than $k$, or
\item $k$ arborescences rooted at $s$ with congestion $n^{o(1)}$, certifying that  $s$-rooted connectivity is at least $k/n^{o(1)}$.
\end{itemize}
The running time is $m^{1+o(1)}$.
\end{restatable}
\Cref{thm:pack} strengthens \Cref{thm:cut} in the senses that it also returns a cerficate of the lower bound for $s$-rooted connectivity, but the approximation factor is $n^{o(1)}$ instead of $\polylog(n)$.




\paragraph{Concurrent Work} Independent of our work, Quanrud \cite{quanrud2025approximatingdirectedconnectivityalmostlinear} and Mosenzon \cite{mosenzon2025almost} have very recently obtained a $(1+\epsilon)$-approximation algorithm for $s$-rooted connectivity running in $m^{1+o(1)}/\epsilon$ time. This $(1+\epsilon)$ guarantee is stronger than our polylogarithmic approximation factor in \Cref{thm:cut}. It remains unclear whether the techniques in \cite{quanrud2025approximatingdirectedconnectivityalmostlinear,mosenzon2025almost} yield any analogous result for arborescence packing, as in \Cref{thm:pack}.

\section{Preliminaries}

Given a directed graph $G=(V,E)$, we denote $n=|V|$ and $m=|E|$. We use $[W]$ to abbreviate the set $\{1, \dots, W\}$. We use the notation $\tilde{O}$ to hide polylogarithmic factors in time complexity.  For any set $S$ and function $f:S \to \mathbb{R}$, unless otherwise stated, for any subset $T \subseteq S$, we define $f(T) = \sum_{t \in T} f(t)$.

\paragraph{Directed Graphs and Connectivity.} In a directed graph $G=(V,E)$, for any set $S \subseteq V$, call $(S,V-S)$ a \emph{cut}. We define $E^+(S)$ as the edges from $S$ to $V \setminus S$, and $E^-(S)$ as the edges from $V \setminus S$ to $S$. We define $\delta(S) = |E^+(S)|$ and $\rho(S) = |E^-(S)|$. When $S=\{s\}$ is a singleton, we abbreviate as $\delta(s) = |E^+(s)|$ and $\rho(s) = |E^-(s)|$. The degree of $s$ is $\deg(s) = \delta(s) + \rho(s)$ and the volume of a set $S$ is $\vol(S) = \deg(S) = \sum_{s \in S} \deg(s)$. We can similarly discuss degree (or volume) for a subgraph $G'$ or an edge subset $X \subseteq E$, which we denote by $\deg_{G'}$ or $\deg_X$, respectively.

The (global) \emph{min-cut} of a directed graph $G$ is defined as $\min_{\emptyset \neq S \subsetneq V} \rho(S)$. Given a fixed source $s \in V$, we can also define the \emph{$s$-rooted min-cut} problem, which asks $\min_{\emptyset \neq S \subseteq V \setminus s} \rho(S)$. Let $G^R$ be the graph obtained by inverting all edges in $G$. One can observe that, the \emph{min-cut} of $G$ can be computed by taking the minimum from both $s$-rooted min-cut in $G$ and $G^R$. Therefore, below we consider the \emph{$s$-rooted min-cut} problem, and we assume $E^-(s) = \emptyset$ by deleting the incoming edges to $s$, which does not influence the result.

\paragraph{Strongly Connected Components.} For a directed graph $G$, we use $[G]$ to denote the collection of strongly connected components (SCCs) in $G$, where each SCC $C$ is represented by a subset of vertices. Therefore, $[G]$ forms a partition of $V(G)$.

\paragraph{Max-Flows and Min-Cuts.} Given a directed graph $G=(V,E)$ and $s,t \in V$, the \emph{max-flow} from $s$ to $t$ is the maximum number of edge-disjoint paths from $s$ to $t$. By duality, this equals the size of \emph{min-cut} from $s$ to $t$, which is $\min_{t \in T \subseteq V-s} \rho(T)$. Here, for the minimizer $T$, we denote the cut $(V-T,T)$ as a \emph{min-cut} from $s$ to $t$. Given a graph $G$ and $s,t$, a max-flow or minimum $(s,t)$-cut can be computed in $m^{1+o(1)}$ time \cite{chen2022maximum}.



\subsection{Arborescences}

For a directed graph $G$ with source vertex $s$, an \emph{arborescence} rooted at $s$ is a spanning tree subgraph of $G$ in which every edge is directed away from $s$. Given an integer $k\geq 1$, a $k$-\emph{arborescence packing} with congestion $\gamma$ is a collection of $k$ arborescences such that each edge appears in at most $\gamma$ arborescences.

\subsection{Directed Expanders}

\paragraph{Demands.} For a directed graph $G$, a \emph{demand} $\mathbf{D}:V(G)\times V(G)\to \mathbb{R}_{\geq 0}$ assigns each (ordered) vertex pair a non-negative real number. The demand $\mathbf{D}$ is
\begin{itemize}
\item \emph{$\mathscr{C}$-component-constrained} for some partition $\mathscr{C}$ of $V(G)$, if for any two vertices $u,v$ belonging to different components $C_{u},C_{v}\in\mathscr{C}$, $\mathbf{D}(u,v) = 0$, and
\item \emph{$\mathbf{d}$-respecting} for some vertex weights $\mathbf{d}:V(G)\to \mathbb{R}_{\geq 0}$, if for each vertex $v$, $\sum_{u}(D(u,v) + D(v,u))\leq \mathbf{d}(v)$.
\end{itemize}

\paragraph{Routings.} For each pair of vertex $(u,v)$, let ${\cal P}_{u,v}$ denote the collection of simple paths from $u$ to $v$. Let ${\cal P} = \bigcup_{u,v\in V(G)}{\cal P}_{u,v}$ denote the collection of all simple paths. A \emph{routing} $\mathbf{R}: {\cal P}\to \mathbb{R}_{\geq 0}$ assigns each simple path a non-negative real number. The routing $\mathbf{R}$ has congestion $\gamma$ if for each edge $e\in E(G)$, $\sum_{P\ni e}\mathbf{R}(P)\leq \gamma$. A demand $\mathbf{D}$ can be routed in $G$ with congestion $\gamma$ if there exists a routing $\mathbf{R}$ with congestion $\gamma$ such that for each vertex pair $(u,v)$, $\sum_{P\in{\cal P}_{u,v}}\mathbf{R}(P) = \mathbf{D}(u,v)$.

\paragraph{Weak Flow Expansion.} For a directed graph $G$ with a terminal edge set $E^{\star}\subseteq E(G)$ and a cut edge set $B\subseteq E(G)$, we say $E^{\star}$ is $[G\setminus B]$-component-constrained $\phi$-expanding in $G$, if any $(\deg_{E^{\star}})$-respecting $[G\setminus B]$-component-constrained demand $\mathbf{D}$ can be routed in $G$ with congestion $1/\phi$.

\begin{lemma}[Flow Expansion Implies Cut Expansion]
\label{lemma:CutExpansion}
If $E^{\star}$ is $[G\setminus B]$-component-constrained $\phi$-expanding in $G$, then for each SCC $C\in[G\setminus B]$ and each 
cut $(V(G)\setminus T, T)$ such that $\deg_{E^{\star}}(C\cap T)\leq \deg_{E^{\star}}(C)/2$, we have
\[
\min\{\delta_{G}(T),\rho_{G}(T)\}\geq \phi\cdot\deg_{E^{\star}}(C\cap T).
\]
\end{lemma}

\subsection{Weighted Graphs} One may also consider \emph{weighted graphs} $G=(V,E,\mathbf{c})$, with $\mathbf{c}: E \to [W]$, where every edge $e$ is with capacity ${\bf c}(e)$. In weighted graphs, instead of letting $\delta(S) = |E^+(S)|$ and $\rho(S) = |E^-(S)|$, we let $\delta(S) = \mathbf{c}(E^+(S))$ and $\rho(S) = \mathbf{c}(E^-(S))$, and the $s$-rooted min-cut is still defined as $\min_{\emptyset \neq S \subseteq V \setminus s} \rho(S)$, which is the minimum \emph{weight} of a cut.

In \Cref{sec:weak-ed}, we use the notion in weighted graphs to conform with \cite{fleischmann2025improved}. In \Cref{sec:apx-cut} and \Cref{sec:apx-pack}, we mainly state the results in unweighted graphs. However, all our algorithms will generalize seamlessly to weighted graphs when $W = \poly(n)$.


\section{Weak Expander Hierarchy}
\label{sec:weak-ed}

In this section, we recall tools related to directed expanders. We will exploit directed expander hierarchies to design our approximate min-cut and arborescence packing algorithms.

\begin{definition}[Expander Hierarchy]
\label{def:hier}
     An $L$-level expander hierarchy $\mathcal{H}$ of $G$ is a collection of edge sets $\{E_{i}\mid 1\leq i\leq L\}$ such that $\bigcup_{i}E_{i} = E(G)$. The hierarchy ${\cal H}$ is $\phi$-\textit{expanding} if, for every level $i \in \{1, \dots, L\}$, the level-$i$ edge set $E_i$ is $\left[G \setminus E_{>i}\right]$-component-constrained $\phi$-expanding in $G$, where $\left[G \setminus E_{>i}\right]$ refers to the set of SCCs in $G \setminus E_{>i}$.
\end{definition}

For intuition, we point out that all SCCs across levels in the hierarchy form a laminar family. Moreover, since we assume that the source vertex $s$ has no incoming edges, $s$ always appears as a singleton SCC in all levels.

\begin{theorem}[\cite{fleischmann2025improved},~Theorem 6.1]
\label{thm:DirectedED}

Given a directed, weighted graph $G = (V, E, \mathbf{c})$ with edge capacities $\mathbf{c}$ bounded by $W$, and a terminal edge set $E^{\star}\subseteq E$, there is a randomized algorithm that with high probability finds a set of cut edges $B\subseteq E$ such that
\begin{enumerate}
    \item $E^{\star}$ is $[G\setminus B]$-component-constrained $\phi$-expanding in $G$ for some $\phi = \Omega(1/(\log n\cdot\log^{3}(nW)))$, and
    \item the cut size satisfies $\mathbf{c}(B)\leq \mathbf{c}(E^{\star})/2$.
\end{enumerate}
The algorithm runs in $m^{1+o(1)}\log^3(nW)$ time.
\end{theorem}

\begin{remark}
\Cref{thm:DirectedED} is not stated identically to Theorem 6.1 in~\cite{fleischmann2025improved}, but the same argument also yields \Cref{thm:DirectedED}.
\end{remark}


\begin{theorem}
\label{thm:ExpanderHierarchy}
    Given a directed, weighted graph $G=(V,E,\textbf{c})$, one can construct a hierarchy $\mathcal{H}$ with $L = O(\log(nW))$ and $\phi = \Omega(1/(\log(n)\log^{3}(nW)))$ in $m^{1+o(1)}\cdot \log^{4}(nW)$ time.
\end{theorem}
\begin{proof}
Initially, we set $E_{1} = E$. For each level $i=1,2,3,...$, we perform the following steps.
\begin{enumerate}
\item Run the directed expander decomposition \Cref{thm:DirectedED} on $G$ with terminal edge set $E_{i}$, and let $E_{i+1}$ be the returned cut edges.
\item If $E_{i+1}$ is empty, we let $L$ be the current $i$ representing the top level, and then exit the loop. Otherwise, proceed to the next iteration.
\end{enumerate}
The hierarchy $\H$ is exactly $\{E_{i}\mid 1\leq i\leq L\}$ computed above.

\medskip

\noindent\underline{Correctness.} 
Consider a level $i$. By \Cref{thm:DirectedED}, $E_{i}$ is $[G\setminus E_{i+1}]$-component-constrained $\phi$-expanding in $G$, for $\phi = \Omega(1/(\log(n)\cdot\log^{3}(nW))$. Because $E_{>i}\supseteq E_{i+1}$, we have that $E_{i}$ is $[G\setminus E_{>i}]$-component-constrained $\phi$-expanding in $G$ as desired. 

The hierarchy has $L = O(\log(nW))$ levels because $\mathbf{c}(E_{i+1})\leq \mathbf{c}(E_{i})/2$ for each $i$, and initially $\mathbf{c}(E_{1}) \leq nW$.



\medskip

\noindent\underline{Running Time.} We run \Cref{thm:DirectedED} for $L = O(\log(nW))$ levels, so the total running time is $m^{1+o(1)}\log^{4}(nW)$.

\end{proof}


\section{Approximate rooted min-cut}
\label{sec:apx-cut}

Given the above hierarchy $\H$ in \Cref{def:hier}, in this section we aim to obtain an $O(L/\phi)$ approximation for the rooted min-cut problem.

\cut*

We describe our algorithm as follows. First, we consider any level $i$ and any component $C$ (except the singleton $\{s\}$) in $[G \setminus E_{>i}]$. Then we consider all edges from $E_i$ inside $C$. We sample $O(\log n)$ times: in each iteration we first choose an edge $e$ uniformly at random from $E_i[C]$, and then choose a random endpoint $v$ of $e$. For any sampled vertex $v$, we compute the min-cut from $E^-(C)$ to $v$, which will correspond to the incoming edges $E^-(C_v)$ for some $C_v \subseteq C$. Finally, we output the minimum value among all computed $\rho(C_v)$. (When $i=0$, every component in $[G-E_{>0}]$ is a singleton $C=\{v\}$ and we still find $C_{v}=C$.) A pseudocode is given in \Cref{alg:apxcut}. 

\begin{algorithm}[ht]
    \caption{Approximate $s$-min-cut}
    \label{alg:apxcut}
    \begin{algorithmic}[1]
        \State $\H \gets$ $\phi$-expanding hierarchy of $G$
        \For{$i = 0$ \textbf{to} L, $C \in [G \setminus E_{>i}]$}
            \State Compute mincut $\rho(C_v)$ from $E^-(C)$ to random $v \gets \text{endpoints of }E_i[C]$
            \State (When $i=0$, $C=\{v\}$ is a singleton and we get $\rho(C_{v})=\rho(v)$.)
            \State Repeat $\Theta(\log n)$ times
        \EndFor
        \State \Return $\min(\rho(C_v))$ for all computed $C_v$
    \end{algorithmic}
\end{algorithm}

Now we prove that the above algorithm will output an $O(L/\phi)$ approximation for the rooted min-cut. Here, we use $\lambda$ to denote the rooted connectivity. From definition, below we suppose $\emptyset \neq T \subseteq V \setminus s$ is a minimizer that $\lambda = \rho(T)$. To prove \Cref{thm:cut}, we aim to find some $C_v$ that we compute in the algorithm satisfying $\rho(C_v) \leq (L/\phi)\lambda = (L/\phi) \rho(T)$. 

Start from $T$, our strategy is to construct one candidate set $T_i \subseteq T$ for each level $i$ that satisfies $\rho(T_i) \leq (L/\phi) \rho(T)$ under some conditions. And then, we show that a candidate set $T_i$ with $\rho(T_i) \leq (L/\phi) \rho(T)$ can be captured by some computed $C_v$ in the algorithm. We show the formal definition of $\{T_i\}$ as follows:

\begin{definition}
    We define $\{T_i\}$ inductively from $T_L = T$. For each $L > i \geq 0$, we consider an arbitrary topological order between the components in $[G-E_{>i}]$. Then we let $C_i$ be the first component in the topological order in $[G-E_{>i}]$ that intersects $T_{i+1}$, and define $T_i = T_{i+1} \cap C_i$.
\end{definition}

We note that $C_i \subseteq C_{i+1}$ for all $i$, since all components in $[G-E_{>i}]$ only further partition all components in $[G-E_{>{i+1}}]$ and thus $T_i = T \cap C_i$. Intuitively, $T_i$ is the intersection of $T$ and the ``first" component $C_i$ that intersects $T$. Next, the following theorem gives a condition when $\rho(T_i) \leq (L/\phi) \rho(T)$ holds. 

\begin{theorem}
    For any layer $i$ in $\mathcal{H}$, 
    suppose that for all $k>i$ we have $\vol_{E_k}(T_k) < |E_k[C_k]| / 2$. Then $\rho(T_i) = O(L\rho(T)/\phi)$.
\end{theorem}

\begin{proof}
    We first claim that, $E^-(T_i) \setminus E^-(T)$ contains only edges in $E_{>i}$. Let's suppose conversely that $E^-(T_i) \setminus E^-(T)$ contains some edge $e \in E_{\leq i}$. Since $E^-(T_i) \setminus E^-(T)$ only contains edges from $T-T_i$ to $T_i$, we suppose $e=(u,v)$ with $u \in T-T_i$ and $v \in T_i$. Now consider the component $C$ in $G-[E_{>i}]$ that intersects $T$ at $u$. Since $C$ has an edge to $C_i$, it must be precedent to $C_i$ in topological order, which is a contradiction by our definition of $C_i$. Thus, we can write as $E^-(T_i) \setminus E^-(T) = \cup_{k>i}E_k^-(T_i)$.

    Next we prove that, for each $k > i$, under the condition that $\vol_{E_k}(T_k) < \vol_{E_k}(C_k) / 2$, we have $|E_k^-(T_i)| \leq \rho(T)/\phi$. One can observe that, suppose $e \in E_k^-(T_i)$, then from the same reasoning as above, both endpoints of $e$ must be in $C_k$. Now from \Cref{def:hier}, we know the edges in $E_k[C_k]$ is $\phi$-expanding. Since $\vol_{E_k}(T_k) < \vol_{E_k}(C_k) / 2$, and $(V-T,T)$ is a cut through $E_k[C_k]$, one can obtain $|E_k^-(T_i)| \leq \rho(T)/\phi$ by \Cref{lemma:CutExpansion}. 
    
    Finally, by summing over all $L$ levels, we have $E^-(T_i) \setminus E^-(T) \leq L\rho(T)/\phi$ and thus $\rho(T_i) \leq (L/\phi + 1)\rho(T) = O(L\rho(T)/\phi)$.
\end{proof}


Now consider the case that $i$ is the largest index that satisfies $\vol_{E_i}(T_i) \geq |E_i[C_i]| / 2$, then $\rho(T_i) = O(L \rho(T)/ \phi)$. Now suppose in the algorithm we have entered the loop with level $i$ and $C=C_i$. When $i=0$, $C$ is a singleton and we simply return $\rho(C) = \rho(T_0) = O(L \rho(T)/\phi)$. 
Otherwise, we show that given $\vol_{E_k}(T_k) \geq \vol_{E_k}(C_k) / 2$, we can capture the set $T_i$ by sampling.

\begin{claim}
    When $C=C_i$, the algorithm samples some $v \in T_i$ with high probability, and $\rho(C_v) \leq \rho(T_i)  = O(L \rho(T)/\phi)$.
\end{claim}

\begin{proof}
    By our sampling method, we will sample each vertex $v$ with probability proportional to its degree in $E_i$. Now since $\vol_{E_k}(T_k) \geq \vol_{E_k}(C_k) / 2$, in every sample, we have at least $1/2$ probability to sample some $v \in T_i$. Therefore, the probability that there exists some $v \in T_i$ in all $\Theta(\log n)$ samples will be at least $1 - (1/2)^{\Theta(\log n)} = 1 - n^{-\Theta(1)}$. When some $v \in T_i$ is sampled, since $(V-T_i,T_i)$ is also a valid cut from $E^-(C)$ to $v$, we can obtain for the min-cut $C_v$ that $\rho(C_v) \leq \rho(T_i)$ and thus $\rho(C_v) = O(L \rho(T)/\phi)$.
\end{proof}

So we conclude that the algorithm will output some $\rho(C
_v) = O(L \rho(T)/\phi)$ with high probability, which gives us an $O(L/\phi)$-approximation. Finally, we show the time complexity of our algorithm and prove \Cref{thm:cut}.

\begin{proof}[Proof of \Cref{thm:cut}]
    From above we know the algorithm will output some $\rho(C_v) = O(L \rho(T)/\phi)$ with high probability. By inserting $L = O(\log n)$ and $\phi = \Omega(1/\log^4n)$, one can get $\rho(C_v) = O(\rho(T)\log^5n)$. For time complexity, the hierarchy $\H$ is built in $m^{1+o(1)}$ time in \Cref{thm:ExpanderHierarchy}. For each layer, all min-cut instances have $\tilde{O}(m)$ edges in total, so the total time of these instances is still $m^{1+o(1)}$. Therefore, the total time for computing $\H$ and resolving $L=O(\log n)$ layers is still $m^{1+o(1)}$.
\end{proof}

\section{Approximate Arborescence Packing}
\label{sec:apx-pack}

In this section, we will present our approximate arborescence packing algorithm, formally stated in \Cref{thm:ArboPackingDetailed}.

\subsection{Section Preliminaries}

\paragraph{Expander Routing.} In \Cref{sec:apx-pack}, we may consider \emph{integral demands}. An integral demand $\mathbf{D}$ is a multiset of vertex pairs. The notions of being $\mathscr{C}$-component-constrained and $\mathbf{d}$-respecting extend to an integral demand $\mathbf{D}$ by interpreting $\mathbf{D}$ as a function that maps each vertex pair to its multiplicity.

An \emph{integral routing} $\mathbf{R}$ of an integral demand $\mathbf{D}$ simply maps each $(u,v)\in \mathbf{D}$ into a $u$-$v$ routing path. $\mathbf{R}$ has congestion $\gamma$ if for each edge $e\in E(G)$, the number of routing paths going through $e$ is at most $\gamma$.

\begin{theorem}[Directed Expander Routing]
\label{conj:ExpanderRouting}
Let $G$ be a directed graph with two edge sets $E^{\star},X\subseteq E(G)$ such that $E^{\star}$ is $[G\setminus X]$-component-constrained $\phi$-expanding in $G$. Given an integral demand $\mathbf{D}$ on $V(G)$ that is $[G\setminus X]$-component-constrained and $(\alpha\cdot\deg_{E^{\star}}^{-})$-respecting for some $\alpha \geq 1$, there is an algorithm that computes an integral routing of $\mathbf{D}$ in $G$ with congestion $\kappa_{\route}\cdot\alpha/\phi$, for some
\[
\kappa_{\route} = n^{o(1)}.
\]
The algorithm runs in $(m+|\mathbf{D}|)\cdot n^{o(1)}$ time.
\end{theorem}

\Cref{conj:ExpanderRouting} is an extension of  almost-linear-time expander routing algorithms in undirected graphs  \cite{ghaffari2017distributed,ghaffari2018new,chang2024deterministic}. We will work on a formal proof of this in the full version of the paper.

\subsection{The Algorithm}

\begin{theorem}
\label{thm:ArboPackingDetailed}
Let $G = (V,E)$ be a directed graph with a given source vertex $s$ and a parameter $k>0$. Assuming \Cref{conj:ExpanderRouting}, there is an algorithm that computes either
\begin{itemize}
\item a $k$-arborescence packing with congestion $\tilde{O}(\kappa_{\route}) = n^{o(1)}$, or
\item a cut $(S,V\setminus S)$ such that $s\in S\subset V(G)$ and $|E^{+}(S)|<k$, certifying that there is no $k$-arborescence packing in $G$.
\end{itemize}
The running time is $m^{1+o(1)}$.
\end{theorem}

To obtain a $k$-arborescence packing, we will actually assign each edge $e\in E(G)$ a subset of \emph{colors} $\Gamma(e)\subseteq [k]$ that satisfies the following.
\begin{enumerate}
\item\label{Prop:FinalColor1} For each color $\gamma\in [k]$, the source $s$ should be able to reach every other vertex using only \emph{$\gamma$-colored} edges (we say an edge is $\gamma$-colored iff $\gamma\in \Gamma(e)$). 
\item\label{Prop:FinalColor2} Each edge $e\in E(G)$ has at most $\tilde{O}(\kappa_{\route})$ colors, i.e. $|\Gamma(e)|\leq \tilde{O}(\kappa_{\route})$. 
\end{enumerate}
Observe that, once we have the above coloring, we can compute $k$ arborescences by for each color $\gamma\in[k]$, picking a DFS tree rooted at $s$ in the subgraph induced by $\gamma$-colored edges. The congestion of these arborescences is clearly bounded by $\tilde{O}(\kappa_{\route})$ because of Property \ref{Prop:FinalColor2}.

Therefore, the task reduces to either compute such a coloring $\{\Gamma(e)\mid e\in E(G)\}$, or output a small cut. To this end, we first construct an expander hierarchy ${\cal H}$ in the graph $G$ with $L$ levels and expansion $\phi$, where
\[
L = O(\log(nW)),\qquad \phi^{-1} = O(\log n\log^{3}(nW)),
\]
using \Cref{thm:ExpanderHierarchy}. With the hierarchy ${\cal H}$, our algorithm will proceed level by level from the bottom to the top.

\subsubsection{Invariants at Level $i$} After the computation at level $i$, our algorithm will keep a color assignment $\Gamma_{i}(e)$ for each edge $e\in E(G)$ and $\Gamma_{i}(v)$ for each vertex $v\in V(G)\setminus\{s\}$. 

The level-$i$ color assignments should satisfy certain invariants. Before describing the invariants, we start with some notations. For each vertex $v\in V(G)$, we use $C_{i,v}$ to denote the unique SCC in $[G\setminus E_{>i}]$ that contains $v$. Moreover, we let $K_{i}(v)$ denote the set of inter-SCC \underline{or} $E_{>i}$ incoming edges of $v$. That is,
\[
K_{i}(v):=\{e\in E(G)\mid e=(u,v)\text{~s.t.~}C_{i,u}\neq C_{i,v}\}\cup E_{>i}^{-}(v),
\]
and we call the edges in $K_{i}(v)$ the \emph{level-$i$ critical incoming edges} of $v$. Let 
\[
\Delta_{i}(v) = |K_{i}(v)|
\]
denote the size of $K_{i}(v)$.
Lastly, the following notions of \emph{$\gamma$-colored} vertices and edges are with respect to the color assignments $\{\Gamma_{i}(e)\}$ and $\{\Gamma_{i}(v)\}$ at level $i$.

The invariants are as follows.
\begin{enumerate}
\item\label{Invariant1} For each vertex $v\in V(G)\setminus \{s\}$ and each color $\gamma\in[k]$, $v$ can be reached from a $\gamma$-colored vertex in $C_{i,v}$ using only $\gamma$-colored edges.
\item\label{Invariant2} For each vertex $v\in V(G)\setminus \{s\}$,  $|\Gamma_{i}(v)|\leq \Delta_{i}(v)\cdot (i+1)$.
\item\label{Invariant3} For each edge $e\in E(G)$,  $|\Gamma_{i}(e)|\leq 5i^{2}\cdot(\kappa_{\route}/\phi)$.
\end{enumerate}
Intuitively, a vertex $v$ having color $\gamma\in\Gamma_{i}(v)$ means that $v$ acts as a ``breakpoint'' in the $\gamma$-colored arborescence, and we want to establish reachability from $s$ to $v$ via upper-level routing.

\paragraph{The Base Case.} We start with the bottom level $i=0$, and initialize level-$0$ color assignments that trivially satisfy the invariants: for each edge $e\in E(G)$, $\Gamma_{0}(e) = \emptyset$; for each vertex $v\in V(G)\setminus\{s\}$, $\Gamma_{0}(v) = [k]$. To see the feasibility of the level-$0$ color assignments, we first observe each SCC in $[G\setminus E_{>0}]$ just contains a single vertex since $E_{>0} = E(G)$. This means $\Delta_{0}(v) = |E^{-}_{G}(v)|$ is just the in-degree of $v$ in $G$.

Invariants \ref{Invariant1} and \ref{Invariant3} clearly hold. For Invariant \ref{Invariant2}, if it does not hold, it means there is a vertex $v$ with $\Delta_{0}(v)<k$, which implies a small cut $(V(G)\setminus \{v\},\{v\})$ and we just terminate the algorithm.

\paragraph{The Ending Case: Getting the Final Coloring $\{\Gamma(e)\}$.} Before we describe how to compute the color assignments level by level, we first show how to get the final edge-coloring $\Gamma(e)$ from the top-level color assignments $\{\Gamma_{L}(e)\}$ and $\{\Gamma_{L}(v)\}$. In fact, this is essentially the same as computing an arborescence packing on a DAG. 

We start by setting $\Gamma(e) = \Gamma_{L}(e)$ for each edge $e\in E(G)$. For each vertex $v\in V(G)\setminus\{s\}$, observe that $K_{L}(v)$ is exactly the inter-SCC incoming edges of $v$. We distribute the colors $\Gamma_{L}(v)$ over the edges in $K_{L}(v)$ such that each edge \emph{receives} at most $L+1$ colors. Here, an edge $e\in K_{L}(v)$ receives a subset of colors $Y'\subseteq \Gamma_{L}(v)$ means that we add $Y'$ into $\Gamma(e)$. This finishes the computation of the final edge-coloring $\{\Gamma(e)\}$.

Property \ref{Prop:FinalColor1} follows Invariant \ref{Invariant1} of the level-$L$ color assignments and the fact that SCCs in $G$ ``form'' a DAG. Each edge $e$ has
\[
|\Gamma(e)|\leq |\Gamma_{L}(e)| + L+1\leq O(L^{2}\cdot (\kappa_{\route}/\phi)).
\]

\subsubsection{The Algorithm at Level $i$} Consider a level $1\leq i\leq L$. We now describe the algorithm to compute the level-$i$ color assignments from the level-$(i-1)$ color assignments. Initially, we set $\Gamma_{i}(v) = \emptyset$ for each $v\in V(G)\setminus\{s\}$ and $\Gamma_{i}(e) = \Gamma_{i-1}(e)$ for each $e\in E(G)$.


\paragraph{Step 1: Allocate ``Breakpoint'' Colors.} For each SCC $C_{i}\in[G\setminus E_{>i}]$ and each vertex $v\in C_{i}$, recall that $C_{i-1,v}$ is the unique SCC in $[G\setminus E_{>i-1}]$ that contains $v$, and clearly $C_{i-1,v}\subseteq C_{i}$. We first partition the level-$(i-1)$ critical incoming edges of $v$, i.e. $K_{i-1}(v)$, into three sets $E_{X}, E_{Y}$ and $E_{Z}$, where
\begin{align*}
E_{X} &= K_{i}(v),\\
E_{Y} &= (E^{-}_{G}(C_{i-1,v})\setminus E^{-}_{G}(C_{i})\setminus E_{>i-1})\cap E^{-}_{G}(v),\\
E_{Z} &= E_{i}^{-}(v)\setminus (E_{X}\cup E_{Y}).
\end{align*}
We make some remarks on this partition. The first set $E_{X}$ is simply the level-$i$ critical incoming edges of $v$, which, by definition, consist of $v$'s incoming edges in $E_{>i}$ and $v$'s incoming edges whose tail-endpoints fall outside $C_{i}$. The second set $E_{Y}$ consists of $v$'s incoming edges whose tail-endpoints fall inside $C_{i}\setminus C_{i-1,v}$, but excludes those inside $E_{> i-1}$. By definition, $E_{X}$ and $E_{Y}$ are clearly disjoint. By the definition of $K_{i-1}(v)$, the remaining edges $K_{i-1}(v)\setminus (E_{X}\cup E_{Y})$ are all inside $E_{i}^{-}(v)$. Therefore, these three sets $E_{X},E_{Y},E_{Z}$ form a partition of $K_{i-1}(v)$.

Next, we arbitrarily partition $v$'s color $\Gamma_{i-1}(v)$ into three sets $X,Y,Z$, such that 
\[
|X|\leq |E_{X}|\cdot i = \Delta_{i}(v)\cdot i,\qquad |Y|\leq |E_{Y}|\cdot i,\qquad |Z|\leq |E_{Z}|\cdot i.
\]
\begin{itemize}
\item (Update $\Gamma_{i}(v)$) For the colors in $X$, we add them to $\Gamma_{i}(v)$. 
\item (Update $\Gamma_{i}$(e)) For the colors in $Y$, we distribute them over edges in $Y$ so that each edge in $Y$ receives at most $i$ colors from $Y$. 
\end{itemize}
The remaining colors in $Z$ will be handled in the following steps, and we use $Z_{v}$ refer to the $Z$ for this vertex $v$.


\paragraph{Step 2: Max Flows.} For each SCC $C_{i}\in[G\setminus E_{>i}]$, we will run a single-commodity max-flow instance in $G[C_{i}]$. The instance has source function $\Delta_{C_{i}} := (\Delta_{i})_{\mid C_{i}}$, i.e., the restriction of $\Delta_{i}$ on vertices $C_{i}$, and sink function $\nabla_{C_{i}}:= (i\cdot \deg^{-}_{E_{i}})_{\mid C_{i}}$, i.e., each vertex $v\in C_{i}$ has $\nabla_{C_{i}}(v)$ equal to $i$ times the number of $v$'s incoming edges in $E_{i}$. Note that, if we let 
\[
Z_{C_{i}} := \bigcup_{v\in C_{i}} Z_{v}
\]
be the union of the $Z$-colors of all $v\in C_{i}$, then the total sink $|\nabla_{C_{i}}| := \sum_{v\in C_{i}}\nabla_{C_{i}}(v)$ satisfies
\[
|\nabla_{C_{i}}|= i\cdot\sum_{v\in C_{i}}|E^{-}_{i}(v)| \geq \sum_{v\in C_{i}}|Z_{v}|\geq |Z_{C_{i}}|
\]
by the definition of $Z_{v}$

\begin{lemma}
\label{lemma:PackingFlow}
By calling an exact max-flow algorithm on $G[C_{i}]$ once, we can obtain either
\begin{itemize}
\item \textbf{(Cut Case)} a non-empty set $C^{\star}\subseteq C_{i}$ such that $|E^{-}(C^{\star})|<k$, or
\item \textbf{(Flow Case)} a collection of $|Z_{C_{i}}|$ edge-disjoint paths in $G[C_{i}]$, denoted by ${\cal P}_{C_{i}}$, such that: for each vertex $u\in C_{i}$, at most $\Delta_{C_{i}}(u)=\Delta_{i}(u)$ paths start at $u$; and for each vertex $v\in C_{i}$, at most $\nabla_{C_{i}}(v)$ paths end at $v$.
\end{itemize}
\end{lemma}
We are not going to prove \Cref{lemma:PackingFlow} formally since it is standard. Roughly speaking, if the max-flow value is less than $\min\{k,|\nabla_{C_{i}}|\}$, then it falls into the cut case. Otherwise, it falls into the flow case by decomposing the max-flow of value at least $\min\{k,|\nabla_{C_{i}}|\}\geq |Z_{C_{i}}|$.

By \Cref{lemma:PackingFlow}, if we obtain a set $C^{\star}$, then $(V(G)\setminus C^{\star},C^{\star})$ is a desired small cut, and thus we can terminate the whole algorithm. Otherwise, we obtain the edge-disjoint paths ${\cal P}_{C_{i}}$. We assign each of them to each color in $Z_{C_{i}}$. 
\begin{itemize}
\item (Update $\Gamma_{i}(e)$) For each color $\gamma\in Z_{C_{i}}$, we assign the color $\gamma$ to each edge $e$ on the corresponding path $P_{C_{i},\gamma}$ (i.e. add $\gamma$ to $\Gamma_{i}(e)$ for each $e\in P_{C_{i},\gamma}$)
\item (Update $\Gamma_{i}(v)$) For the starting vertex of $P_{C_{i},\gamma}$, say $u$, we add the color $\gamma$ to $\Gamma_{i}(u)$.
\end{itemize}
Lastly, we refer to the ending vertex of $P_{C_{i},\gamma}$ as the \emph{$\gamma$-colored leader} of $C_{i}$, denoted by $w_{C_{i},\gamma}$.


\paragraph{Step 3: Expander Routing.} For each SCC $C_{i}\in[G\setminus E_{>i}]$ and each color $\gamma\in Z_{C_{i}}$, we refer to a vertex $v\in C_{i}$ with $\gamma\in Z_{v}$ as a \emph{$\gamma$-colored breakpoint}, and we let 
\[
Z^{-1}_{C_{i}}(\gamma) = \{v\in C_{i}\mid \gamma\in Z_{v}\}
\]
collect all $\gamma$-colored breakpoints of $C_{i}$. Then, we fix an arbitrary \emph{chain demand} $\mathbf{D}_{C_{i},\gamma}$ over the $\gamma$-colored leader $w_{C_{i},\gamma}$ and breakpoints $Z^{-1}_{C_{i}}(\gamma)$ of $C_{i}$. That is, pick an arbitrary order of vertices in $Z^{-1}_{C_{i}}(\gamma)$, say $v_{1},v_{2},...,v_{q}$ (where $q = |Z^{-1}_{C_{i}}(\gamma)|$), and then create $q$ demand pairs $(w_{C_{i},\gamma},v_{1})$, $(v_{1},v_{2})$, ..., $(v_{q-1},v_{q})$.

Let $\mathbf{D}_{i} = \sum_{C_{i},\gamma}\mathbf{D}_{C_{i},\gamma}$ be the total demand set. The following \Cref{lemma:RepsectingDemand} enables us to route $\mathbf{D}_{i}$ via expander routing.

\begin{lemma}
\label{lemma:RepsectingDemand}
The total demand $\mathbf{D}_{i}$ is $(3i\cdot \deg^{-}_{E_{i}})$-respecting and $[G\setminus E_{>i}]$-component-constrained.
\end{lemma}
\begin{proof}
First, $\mathbf{D}_{i}$ is clearly $[G\setminus E_{>i}]$-component-constrained by definition. To see that $\mathbf{D}_{i}$ is $(3i\cdot \deg^{-}_{E_{i}})$-respecting, we will show that each vertex $v\in V(G)$ appears in at most $3i\cdot |E^{-}(v)|$ demand pairs. First, $v$ can serve as the $\gamma$-colored leaders of $C_{i,v}$ for at most $\nabla_{C_{i}}(v) = i\cdot |E^{-}(v)|$ colors by \Cref{lemma:PackingFlow}. This contributes at most $i\cdot |E^{-}_{i}(v)|$ appearances. Second, $v$ is a $\gamma$-colored breakpoint of $C_{i,v}$ for each color $\gamma\in Z_{v}$ (note that $|Z_{v}|\leq i\cdot |E^{-}_{i}(v)|$), and each breakpoint contributes two appearances.
\end{proof}

By \Cref{lemma:RepsectingDemand} and \Cref{conj:ExpanderRouting}, we can obtain a routing $\mathbf{R}_{i}$ of $\mathbf{D}_{i}$ in $G$ with
\[
\text{congestion~}(3i)\cdot\kappa_{\route}/\phi.
\]
Therefore, for each chain demand $\mathbf{D}_{C_{i},\gamma}$ and each demand pair $(w_{1},w_{2})\in \mathbf{D}_{C_{i},\gamma}$, we have an associated routing path $R$ in $G$ from $w_{1}$ to $w_{2}$. 
\begin{itemize}
\item (Update $\Gamma_{i}(e)$) For each edge $e$ on $R$, we add color $\gamma$ to $\Gamma_{i}(e)$.
\end{itemize}
This finishes the computation of $\{\Gamma_{i}(e)\mid e\in E(G)\}$ and $\{\Gamma_{i}(v)\mid v\in V(G)\setminus\{s\}\}$.


\paragraph{Proof of the Invariants.} We now show that the level-$i$ color assignments satisfy all the invariants.

\medskip

\noindent\underline{Invariant \ref{Invariant1}.} First of all, observe that SCCs $C_{i-1}\in [G\setminus E_{>i-1}]$ admit a \emph{topological order} $\pi_{i-1}$, such that for each edge $e=(u,v)\in E(G)\setminus E_{>i-1}$ connecting two different SCCs $C_{i-1,u}$ and $C_{i-1,v}$, the SCC $C_{i-1,u}$ precedes $C_{i-1,v}$ on $\pi_{i-1}$. Moreover, the SCC $C_{i-1,s}$ containing the source vertex $s$ is the first in $\pi_{i-1}$.

Now, consider a vertex $v\in V(G)\setminus\{s\}$ and a color $\gamma\in[k]$. To avoid clutter, let $C_{i}$ be the SCC $C_{i,v}\in[G\setminus E_{>i}]$ containing $v$. We want to show that $v$ can be reached from a $\gamma$-colored (w.r.t. $\Gamma_{i}(\cdot)$) vertex in $C_{i}$, using only $\gamma$-colored edges. By the invariants of the level-$(i-1)$ color assignments, we know that $v$ can be reached by a vertex $u\in C_{i-1,v}$ s.t. $\gamma\in \Gamma_{i-1}(u)$ using only $\gamma$-colored edges. Recall that in Step 1, we partition $\Gamma_{i-1}(u)$ into three sets $X_{u},Y_{u},Z_{u}$. 
\begin{itemize}
\item Suppose $\gamma\in X_{u}$. Then according to the update of $\Gamma_{i}(u)$ in Step 1, $\gamma$ is added to $\Gamma_{i}(u)$, and we are done.
\item Suppose $\gamma\in Z_{u}$. Following the routing paths of the chain demand $\mathbf{D}_{C_{i},\gamma}$, $u$ can be reached from the $\gamma$-colored leader $w_{C_{i},\gamma}$ of $C_{i}$ using only $\gamma$-colored edges. Moreover, following the flow path $P_{C_{i},\gamma}$, $w_{C_{i},\gamma}$ can be reached from the starting vertex of $P_{C_{i},\gamma}$ (which is a $\gamma$-colored vertex) using only $\gamma$-colored edges.
\item Suppose $\gamma\in Y_{u}$. In Step 1, $\gamma$ is assigned to an incoming edge $(v',u)$ of $u$ such that $v'\in C_{i}\setminus C_{i-1,v'}$, and this edge $(v',u)$ does not belong to $E_{>i-1}$. This time, we repeat the above argument for the vertex $v'$ and the color $\gamma$. Note that we make progress each time, as $C_{i-1,v'}$ precedes $C_{i-1,v}$ in the order $\pi_{i}$, so finally we will reach the above two cases which find a $\gamma$-colored vertex in $C_{i}$ that can reach $v$ using only $\gamma$-colored edges.
\end{itemize}

\medskip

\noindent\underline{Invariant \ref{Invariant2}.} Consider a vertex $v\in V(G)\setminus\{s\}$. Initially, $\Gamma_{i}(v)$ is empty. In Step 1, the size of $\Gamma_{i}(v)$ increases by at most $\Delta_{i}(v)\cdot i$. In Step 2, the size of $\Gamma_{i}(v)$ increases by at most $\Delta_{i}(v)$ because of \Cref{lemma:PackingFlow}. Hence $|\Gamma_{i}(v)|\leq \Delta_{i}(v)\cdot(i+1)$ at the end.

\medskip

\noindent\underline{Invariant \ref{Invariant3}.} Consider an edge $e\in E(G)$. Initially, $|\Gamma_{i}(e)| = |\Gamma_{i-1}(e)|\leq 5(i-1)^{2}\cdot\kappa_{\route}/\phi$.
In Steps 1, 2 and 3, $|\Gamma_{i}(e)|$ increases by at most $i$, $1$ and $(3i)\cdot\kappa_{\route}/\phi$ respectively. Therefore,
\[
|\Gamma_{i}(e)|\leq 5(i-1)^{2}\cdot\kappa_{\route}/\phi + i+1+(3i)\cdot\kappa_{\route}/\phi\leq 5i^{2}\cdot\kappa_{\route}/\phi.
\]

\subsubsection{Running Time Analysis}

The bottlenecks are the expander hierarchy algorithm (\Cref{thm:ExpanderHierarchy}), exact max-flow algorithm (\Cref{lemma:PackingFlow}), and the expander routing (\Cref{conj:ExpanderRouting}). All of them incur $n^{o(1)}$ overheads, which dominate other polylogarithmic overheads. Hence the total running time is $m^{1+o(1)}$.

\bibliographystyle{alpha}
\bibliography{ref}

\end{document}